\documentclass[12pt]{article} \usepackage[utf8]{inputenc}
\usepackage{fullpage}
\usepackage{amssymb} \usepackage{amsfonts} \usepackage{amsthm}
\usepackage{amsmath} \usepackage{listings} \usepackage{xcolor}
\usepackage{lineno} \usepackage[linesnumbered, ruled,
  vlined]{algorithm2e}

\theoremstyle{definition} \newtheorem{definition}{Definition}
\newtheorem{example}{Example} 
\newtheorem{proposition}{Proposition} \newtheorem{remark}{Remark}
\newtheorem{corollary}{Corollary} \newtheorem{theorem}{Theorem}

\newcommand{\R}{\mathbb{R}}

\title{New Algorithms for Solving Tropical Linear Systems}
\author{Alex Davydow} \date{\today}

\begin{document}

\maketitle

\begin{abstract}
The problem of solving tropical linear systems, a natural problem of
tropical mathematics, has already proven to be very interesting from
the algorithmic point of view: it is known to be in $NP\cap coNP$ but
no polynomial time algorithm is known, although counterexamples for
existing pseudopolynomial algorithms are (and have to be) very
complex.

In this work, we continue the study of algorithms for solving tropical
linear systems.  First, we present a new reformulation of Grigoriev's
algorithm that brings it closer to the algorithm of Akian, Gaubert,
and Guterman; this lets us formulate a whole family of new algorithms,
and we present algorithms from this family for which no known
superpolynomial counterexamples work. Second, we present a family of
algorithms for solving overdetermined tropical systems. We show that
for weakly overdetermined systems, there are polynomial algorithms in
this family. We also present a concrete algorithm from this family
that can solve a tropical linear system defined by an $m\times n$
matrix with maximal element $M$ in time $\Theta\left({m \choose n}
\mathrm{poly}\left(m, n, \log M\right)\right)$, and this time matches
the complexity of the best of previously known algorithms for feasibility
testing. 
\end{abstract}

\section{Introduction}

\subsection{Tropical mathematics and tropical linear algebra}

\emph{Tropical mathematics} unites three closely connected fields of
study: tropical algebra, tropical analysis, and tropical geometry. The
term is usually taken to mean mathematics obtained from classical
mathematics by replacing the addition and multiplication operations
with minimum and addition respectively, hence the term \emph{min-plus}
algebra. Sometimes maximum is used instead of minimum, with perfectly
symmetrical results, so in what follows we always use the minimum
operation. Taking the minimum in tropical context is usually denoted
by $\oplus$; addition, by $\otimes$. The $\oplus$ operation is
idempotent, i.e., $a \oplus a = a$, so tropical mathematics is in fact
a part of idempotent mathematics, although lately these notions have
often been identified so that the term ``tropical'' is sometimes
applied to any mathematical constructions with an idempotent
operation.

Tropical algebra was the first section of tropical mathematics to
appear. The term originates from French mathematicians and first
appears in the 1980s. Although different authors attribute the term to
different researchers \cite{Sturmfels04, Pin98, Litvinov07}, all
sources agree that the term came into general use in the honor of one
of the founders of this field, a Brazilian mathematician Imre Simon,
so the term ``tropical'' simply means how French mathematicians viewed
Brazil. Simon himself uses the term already in the work \cite{Imre88}
that laid out the foundations of tropical algebra, attributing the
term to Christian Choffrut.

At first, the term ``tropical'' was used for the discrete version of
$\left( \min, + \right)$ algebra, but at present the terminology has
shifted, and tropical algebra is usually meant to be an algebra over
the semifield $\R_{min}$ (as stated above) or even, sometimes, an
arbitrary algebra with an idempotent operation, e.g., $\left( \min,
\max\right)$ algebra.  Both Simon and his French colleagues used
tropical algebra for the study of finite state machines.

Although a systematic study of the tropical semiring began only after
the works of Simon, we should note that the $\left( \min, + \right)$
semiring had appeared before in optimization problems. For instance,
Floyd's algorithm for finding shortest paths in a graph that was
proposed in the 1960s \cite{Floyd62} can be considered as taking a
tropical degree of the distance matrix (tropical exponentiation is
similar to classical with the difference that we replace addition with
$\oplus$ and multiplication with $\otimes$).  Speaking about
idempotent algebra as a whole, the first work to make serious use of
an algebra over an idempotent ring (apart from Boolean fields) was the
work of Kleene \cite{Kleene56} that studied nerve nets in the context
of finite state machines.  At present, there is a host of literature
of matrices with idempotent coefficients and their applications, e.g.,
\cite{Cuninghame-Green79, Krivulin09, Litvinov11}.

\emph{Tropical linear algebra} is a subfield of tropical algebraic
geometry that works with systems of tropical linear
equations. Unfortunately, definitions remain a major problem in this
subfield. Many notions of linear algebra have several equivalent
definitions, but after tropicalization equivalence disappears, and we
are left with several different definitions. One striking example of
this phenomenon is that there may be several conflicting definitions
for the root of a tropical polynomial.

In the case of the root of a polynomial, researchers finally settled
on the definition of O. Viro, and now the set of roots of a polynomial
are the set of this polynomial's non-smoothness points. This
definition prevailed because it preserves such important properties
as, for instance, the fundamental theorem of algebra (that a
polynomial of degree $n$ has exactly $n$ roots, counting
multiplicities).

The problem of solving linear systems was formulated right after the
definition of a root for a tropical polynomial was given, but the
first work actually devoted to tropical linear algebra appeared only
as late as 2005 \cite{Develin05}. At present, this field is primarily
being developed in France (Akian, Gaubert, Grigoriev and others),
sometimes in collaboration with researchers from other countries
(Izhakian, Guterman).

Since there are no known efficient algorithms for the main problems of
tropical linear algebra, it is currently little used in practice.
Nevertheless, there are practical problems that would benefit from
developments in this field. For instance, Noel, Grigoriev, Vakulenko,
and Radulescu have recently proposed a way to use algorithms for
solving tropical linear systems to study stable states of reaction
networks in biology \cite{NGVR12a, NGVR12b}. Thus, problems of
tropical linear algebra are important from both theoretical and
practical points of view.

It is worth noting that many definitions and theorems of tropical
geometry (including, in particular, tropical linear algebra) appeared
much earlier than tropical geometry itself took shape as a field of
study. In these cases, tropical geometry serves as a language in which
it is convenient to state theorems that have already been proven from
a different viewpoint. Although at first glance this serves little
useful purpose, sometimes such a translation may lead to new
results. For instance, a translation of the Viro's patchworking method
to the language of tropical geometry has led Mikhalkin to an algorithm
for computing Gromov--Witten invariants \cite{Mikhalkin03}.

\subsection{Tropical linear systems}

After researchers had agreed on the definitions of a tropical
polynomial and tropical root, one of the first problems that they
tried to solve was constructing an algorithm for testing the
feasibility of a linear tropical system. However, unlike the classical
case the tropical problem turned out to be much harder, and despite
the fact that it was put forward five years ago, no efficient solution
is known to date.

Similar to the classical case, a tropical linear system can be
conveniently defined with a matrix. One way to test the feasibility of
a classical linear system is to test whether its determinant is
zero. Therefore, due to the idempotent correspondence principle, one
could expect something similar in the tropical case as well. A formula
for the tropical determinant was first proposed by Izhakian in 2008
\cite{Izhakian08}. The determinant was defined completely similar to
the classical determinant with the sole difference that one uses
$\oplus$ instead of addition, $\otimes$ instead of multiplication, and
there is no $(-1)^n$ factor (i.e., this construction corresponds to
the determinant and the permanent at the same time). Izhakian also
showed that a tropical system defined by a square matrix has a
solution if and only if its coefficients are not roots of the tropical
determinant (i.e., the determinant of this matrix is minimized only on
a single monomial). Such square matrices were called \emph{tropical
  singular}.  Interestingly, 14 years before Izhakian, Butcovič
already defined tropical singularity under the name of strong
regularity \cite{Butcovic94}.  However, it was done before tropical
geometry appeared in earnest, and Butcovič did not establish any
connection between strong regularity and feasibility of tropical
systems.

It is important to note that Butcovič proposed an efficient algorithm
for testing a tropical matrix for singularity \cite{Butcovic94}.  One
can also note that the singularity condition is equivalent to the
existence of a unique minimal weight matching and can therefore be
efficiently tested with, e.g., the Hungarian method
\cite{Kuhn55}. Thus, the feasibility problem for tropical linear
systems defined by square matrices was efficiently solved by Izhakian
in 2008. However, for other matrices even pseudopolynomial algorithms
(that would work in time polynomial of the numerical value of the
input rather than its size) were not known.

In 2009, Izhakian generalized his definition to rectangular matrices
\cite{Izhakian09} and showed that a system defined by a rectangular
matrix is infeasible if and only if it contains a singular submatrix
of maximal width (i.e., width equal to the width of the matrix).  Note
that Izhakian's results imply that similar to the classical case, a
tropical system with fewer equations than variables is always
feasible. However, since a rectangular matrix may contain an
exponential number of square submatrices of maximal width (in case
when the height is much larger than the width), this result did not
imply an efficient algorithm for solving tropical linear systems
defined by rectangular matrices. We should also note that in the case
of finite coefficients Izhakian's results are in fact a simple
corollary of the theorem that establishes that Kapranov rank and
tropical rank of a matrix are maximized simultaneously, a theorem
proven by Develin, Santos, and Sturmfels in 2005 \cite{Develin05}.

In case when the number of equations and variables coincide, one can
not only efficiently test the system for feasibility but also solve
it. To do so, Grigoriev proposed\cite{Grigoriev-personal} to drop the
equation that intersects with minimal matchings in two cells and apply
the tropical Cramer's rule to the rest of the matrix
\cite{Gebert03}. Thus, the problem of solving a tropical
system with a square matrix can be solved in polynomial time.

The first pseudopolynomial algorithm for solving rectangular matrices
was presented by Akian, Gaubert, and Guterman in 2010 \cite{Akian10}.
They showed that the feasibility problem for a tropical linear system
has a polynomial reduction to the problem of finding the winner in
mean payoff games. Mean payoff games were proposed by Ehrenfeucht and
Mycielski in 1979 \cite{Erenfeucht79} and have already been quite
comprehensively studied by 2010. For instance, in 1993 Karzanov and
Lebedev, using the results of Karp \cite{Karp78}, showed that the
problem of finding the winner in a mean payoff game lies in the
intersection of complexity classes $NP$ and $coNP$ \cite{Karzanov93};
in 1995, Zwick and Paterson proposed a pseudopolynomial algorithm for
solving this problem \cite{Zwick95}.  Several times, there appeared
algorithms that claimed to find the winner of a mean payoff game in
polynomial time \cite{Karzanov88, Lozovanu91, Lozovanu93}, but so far
all of them turned out to contain mistakes \cite{Zwick95}.

Thus, it was shown in 2010 that the feasibility problem for a tropical
system lies in the intersection of $NP$ and $coNP$; this is an
interesting complexity class: there are few problems known to be in
$NP \cap coNP$ but not known to be in $P$, and most of them are
polynomially equivalent to each other. The feasibility problem also
got a pseudopolynomial algorithm. However, it is easy to construct an
example of a matrix that results, by the algorithm of Akian, Gaubert,
and Guterman, in a mean payoff game with no known efficient
algorithm. Therefore, the problem of finding an efficient algorithm
for testing feasibility of a tropical linear system remained open.

In 2011, Grigoriev proposed a different pseudopolynomial algorithm
\cite{Grigoriev13} similar to the Gram--Schmidt process that is
designed to immediately solve the feasibility problem for a tropical
system. Apart from its relative simplicity, an important feature of
Grigoriev's algorithm was that apart from the pseudopolynomial
estimate he immediately found an exponential complexity bound that
does not depend on the numbers that occur in the matrix. Later Davydow
showed \cite{Davydow13eng} that if one fixes any one of the three
parameters (width, height, and maximal coefficient in the matrix), the
algorithm will work in time polynomial with respect to the other two
parameters. These were optimistic results: they meant that if there
were examples of inputs on which Grigoriev's algorithm is not
polynomial, they would have to be very complex since all three
parameters would have to change simultaneously. Unfortunately, in 2012
Davydow found such a series \cite{Davydow13eng}, and by the time
Grigoriev's work was published it was already known that this
algorithm is also not efficient in the general case.

As we have already mentioned, in 2010 Akian et al. \cite{Akian10}
showed that the feasibility problem for a tropical linear system can
be polynomially reduced to the problem of finding the winner in mean
payoff games. In 2012, Grigoriev and Podolskii \cite{Grigoriev12}
found the inverse reduction, showing that the feasibility problem for
a tropical linear system and the winner problem in a mean payoff game
are equivalent. They constructed a reduction to the problem of
$\left(\min, +\right)$ systems (systems of equations of the form
$Ax=Bx$, where matrices are multiplied by vectors in the tropical
sense, and the equivalence of the problem of $\left(\min, +\right)$
systems to the problem of mean payoff games was proven by Bezem,
Nieuwenhuis, and Rodrígez-Carbonell already in 2010 \cite{Bezem10}.
At the same time, a similar result was directly obtained by Akian,
Gaubert, and Guterman \cite{Akian12}.

On one hand, after Grigoriev and Podolskii showed that testing
feasibility for a tropical system is as hard as finding the winner in
a mean payoff game, it became clear that finding an efficient
algorithm is rather unlikely. On the other hand, they showed that this
problem is interesting not only as an independent problem but also as
a completely new approach to mean payoff games, and that this problem
deserves an even more detailed scrutiny.

Thus, at present there exist efficient feasibility testing algorithms
for systems where the number of equations exceeds the number of
variables by a predefined constant; we will call such systems
\emph{weakly overdetermined}. In this case, it suffices to enumerate
all square submatrices of maximal width for the system's matrix and
test each of them for singularity with the Hungarian method. Then,
Izhakian's results imply that a system if feasible if and only if all
resulting matrices are nonsingular. For systems with an unbounded
number of equations, only pseudopolynomial algorithms are known
(algorithm of Akian--Gaubert--Guterman and Grigoriev's algorithm).
Note that in case of systems with bounded number of variables
Grigoriev's algorithm works in polynomial time.

\subsection{Our contributions}

In this work, we present two main results related to new algorithms
for solving tropical linear systems. First, we present a new
description for Grigoriev's algorithm for solving tropical linear
systems that lets us generalize Grigoriev's algorithm and another well
known algorithm of Akian, Gaubert, and Guterman and consider a whole
family of algorithms that differ in the lifting operation. Even a
simple straightforward combination of these two algorithms can already
solve hard counterexamples for both Grigoriev's and
Akian--Gaubert--Guterman algorithms, and we leave devising hard
counterexamples for these new algorithms as an interesting open
problem.

Second, we present a new algorithm for overdetermined tropical linear
systems, i.e., systems that have more equations than variables. We
reduce solving an overdetermined tropical system to solving several of
its subsystems. This leads to a general algorithm that works on all
tropical systems and runs in time $\Theta\left({m \choose n}
\mathrm{poly}\left(m, n, \log M\right)\right)$.  Moreover, we show
that weakly overdetermined tropical systems (where equations outnumber
variables only by a predefined constant) admit a polynomial time
solution.
 
The paper is organized as follows. In Section~\ref{sec:statement}, we
give a formal definition of the feasibility problem for a tropical
system and introduce the notation used throughout the
paper. Section~\ref{sec:grigoriev} is devoted to optimizing
Grigoriev's algorithm for solving linear tropical systems; our
optimization leads to a unified approach for Grigoriev's algorithm
together with the algorithm of Akian, Gaubert, and Guterman, which in
turn lets us combine the two algorithms, getting an algorithm with no
known counterexamples where it would have to work for superpolynomial
time. Finally, as we have already mentioned, it is possible to test
feasibility of weakly overdetermined tropical linear systems in
polynomial time. Section~\ref{sec:overdet} presents a novel algorithm
that actually solves such systems in polynomial
time. Section~\ref{sec:conclusion} concludes the paper.

\section{Problem setting}\label{sec:statement}

In this section, we give basic definitions regarding tropical linear
systems.

\begin{definition}
A \emph{tropical linear system} is a rectangular matrix of size $m
\times n$. A \emph{solution} of a tropical linear system is a row of
$n$ elements such that after adding it to each row of the matrix each
sum does not contain a strict minimum, i.e., the minimal element
occurs at least twice in every row. A tropical linear system is called
\emph{feasible} if there exists a row that is a solution of this
system \cite{Grigoriev13}.
\end{definition}

\begin{example}
For the matrix
  \[ \left( \begin{array}{ccc}
  1 & 2 & 3 \\ 3 & 2 & 1
  \end{array} \right), \]
the row
  \[ \left( \begin{array}{ccc} 
    1 & 0 & 1
  \end{array} \right) \]
represents a solution: after adding it to the first row we get the
minimal value $2$ in the first and second columns; for the second row,
we get the minimal value $2$ in the second and third columns.
\end{example}

\begin{example}
The matrix
  \[ \left( \begin{array}{cc}
  1 & 2 \\ 3 & 2
  \end{array} \right) \]
is obviously infeasible.
\end{example}

In this work we will consider the problem of establishing feasibility
for \emph{integer-valued} tropical linear systems.  First, we note
that we can apply some simple transformations to a system's matrix
without changing its feasibility status.  The following proposition is
obvious.

\begin{proposition}\label{transformation}
The class of feasible tropical linear systems is invariant with
respect to adding an arbitrary constant to all numbers in one row or
in one column. Moreover, given a solution of the system after such a
transformation, one can find a solution of the original system by
adding to the solution the difference between first rows of the matrix
before and after the transformation.
\end{proposition}

Proposition~\ref{transformation} immediately implies the following
remark.

\begin{remark}
Without loss of generality we can assume that all elements of the
system's matrix are nonnegative.
\end{remark}

In what follows we introduce the following notation for a matrix $A$:
\begin{itemize}
\item $m(A)$, the number of rows in the matrix $A$,
\item $n(A)$, the number of columns in the matrix $A$,
\item $k(A) = n(A) - m(A)$,
\item $M(A)$, the maximal number in the matrix $A$,
\item $R(A)$, the set of rows of the matrix $A$.
\item $a_i(A)$, the $i^{\text{th}}$ row of the matrix $A$.
\end{itemize}
We will omit the argument in this notation if it is clear from context
what matrix we are talking about.

\section{Grigoriev's algorithm and its modifications}\label{sec:grigoriev}

\subsection{The original algorithm}
One recently proposed algorithm for solving tropical linear systems is
Grigoriev's algorithm. As we have already noted, a key feature of this
algorithm is that the work \cite{Grigoriev13} that proposes this
algorithm immediately shows both an upper bound on the algorithm's
complexity that polynomially depends on the matrix size (but it is
polynomial in $M$, the largest element of the matrix, rather than
$\log{M}$) and an upper bound that polynomially depends on $\log{M}$
(but it is not polynomial in matrix size).

We begin with a description of Grigoriev's algorithm. We begin by
noting that due to Proposition~\ref{transformation} we can find not a
solution of the matrix but rather a series of transformations that
consists of adding a constant to all elements in a row or in a column
that would reduce the original matrix to a matrix that has a zero row
for a solution. In what follows, we call such a matrix the
\emph{solution matrix}; finding it is equivalent to finding a
solution.

To solve a system of size $m~\times~n$, we proceed by induction and
assume that we have solved the system of size $(m - 1)~\times~n$
obtained from the initial system by removing its first row. From this
moment on we will assume that all rows of the matrix, except possibly
the first row, do not contain strict minima. Next we define the
\emph{lifting} operation with Algorithm~\ref{Lift1}.

\begin{algorithm}
  \SetAlgoLined 
  \DontPrintSemicolon
  \caption{Matrix lifting in Grigoriev's algorithm\label{Lift1}} 
  \KwData {a matrix $A$ and the
    index $i$ of the column where the first row's minimum is located.}
  \KwResult {if the lifting is possible then $A$ is the lifted matrix}
  $J \leftarrow \{i\}$\; 
  \While{there exists a row in which exactly
    one minimum is achieved in column $j$ such that $j \notin J$} { 
    $J \leftarrow J \cup \{j\}$\; 
  }
  \uIf{$\left|J\right| = n(A)$}{ 
    output that lifting is impossible\; 
  } 
  \Else { 
    $a \leftarrow \infty$\;
    \For{$i = 0$ \KwTo $m$} { 
      $a_i \leftarrow $ maximal number one can
      add to columns with indices from the set $J$ in such a way that
      minimal elements remain minimal in the row with index $i$\; $a
      \leftarrow \min(a, a_i)$\; 
    } 
    \ForEach{$i \in J$} { 
      add $a$ to the elements of column $i$\; 
    } 
  }
\end{algorithm}

Note that although the lifting algorithm does contain some
indeterminacy: it is not specified in what order we add columns to the
set $J$, when the first loop ends the set $J$ is uniquely defined
since the maximal by inclusion such set is unique.

\begin{algorithm}
  \SetAlgoLined \DontPrintSemicolon
  \caption{Grigoriev's algorithm \label{Grigoriev}} \KwData {$A$, a
    matrix of the tropical system} \KwResult {if the tropical system
    defined by $A$ was feasible then its solution matrix, else
    ``infeasible''} Run this algorithm for the matrix $A'$ resulting
  from $A$ by deleting the first row.\; Reduce the matrix $A$ to such
  a form that no row except possibly the first contains a strict
  minimum.\; \While{ the first row contains a strict minimum } { \uIf
    {lifting of the matrix $A$ is impossible} { \Return
      ``infeasible''; } \Else { lift matrix $A$.  } } \Return $A$.
\end{algorithm}

Then we transform the matrix according to Algorithm~\ref{Grigoriev}.
The time complexity of this algorithm was originally shown to be
$O(m^2 n^2M\log{M})$ \cite{Grigoriev13}. Later, a different estimate
of $O(\log{M} \cdot m \cdot n^2 \cdot {{m + n} \choose n} )$ was shown
by Davydow \cite{Davydow13eng}.  Besides, there is a known
counterexample for Grigoriev's algorithm: there exists a sequence of
matrices (with unbounded growth in the number of rows and columns)
such that Grigoriev's algorithm takes $\Omega(n^{\frac{m}{6}}\log{M})$
time to process them, where $M = \mathrm{poly}(n^{\frac{m}{6}})$
\cite{Davydow13eng}.

\subsection{Properties of the solutions found by Grigoriev's algorithm}

We introduce a partial ordering on the solutions of a tropical linear
system: we say that one solution is less than the other if it is
smaller componentwise. Then the following theorem holds.

\begin{theorem}\label{thm:smallest}
For a matrix with one strict minimum, Grigoriev's algorithm finds the
smallest nonnegative solution. In terms of solution matrices,
Grigoriev's algorithm finds the smallest solution matrix which is
greater than the original matrix.
\end{theorem}

\begin{proof}
First note that the existence of such a solution follows from the
well-known fact that the set of solutions is linear.

We prove this theorem by induction on the number of liftings. Namely,
we show that on each step of Grigoriev's algorithm the matrix does not
become less than the maximal solution matrix among those that are
smaller than the original matrix. For the induction base, note that
the original matrix obviously satisfies this condition.

For the induction step, note that if, during a lifting, we add to at
least one column a number smaller than the one added in the algorithm,
then the column to which we added the smallest number will have a
strict minimum.  This means precisely that in the smallest solution
matrix among those that larger than the original matrix we have to add
at least as much as Grigoriev's algorithm adds.
\end{proof}

\subsection{Optimizing Grigoriev's algorithm}

We begin with a simple corollary of Theorem~\ref{thm:smallest}.

\begin{corollary}\label{cor:interrupt}
If Grigoriev's algorithm has changed every column at least once, it
will output ``infeasible''.
\end{corollary}
\begin{proof}
In the smallest solution, at least one of the elements must be zero;
otherwise, one could subtract it from every element and get a smaller
solution.
\end{proof}

Corollary~\ref{cor:interrupt} implies our first optimization of
Grigoriev's algorithm: we can interrupt it and output ``infeasible''
not when lifting is impossible, but rather when each column has been
changed at least once, which can happen much earlier.

For a second optimization, we can also do without the recursion on the
matrix height: we can simply add all columns with a strict minimum to
the set $J$ from the very beginning. Note that while the first
optimization obviously cannot hurt Grigoriev's algorithm, we do not
know this for the second idea, although we have failed to find an
example where the original version of the algorithm would work faster
than this modification.

For this version of Grigoriev's algorithm, the same upper and lower
bounds can be proven in the same way as for the original version.
Apart from some simplification, this version of Grigoriev's algorithm
has the advantage that it is now very similar to the
Akian--Gaubert--Guterman algorithm; the only difference remains in the
lifting operation: the Akian--Gaubert--Guterman algorithm lifts only
columns with strict minima, and only for the value needed in order for
the minima to cease being strict (see Algorithm~\ref{alg:agglifting}).

\begin{algorithm}
  \label{alg:general}
  \SetAlgoLined 
  \DontPrintSemicolon
  \caption{General scheme} 
  \KwData {$A$, a matrix of the tropical system} 
  \KwResult {if the tropical system defined by $A$ was
    feasible then its solution matrix, else ``infeasible''}
  \While { there exist rows with strict minima, and there exists a
    column that has not been lifted} { 
    $A \leftarrow \mathrm{Lifting}(A)$ 
  } 
  \If {there exist rows with strict minima} {
    \Return ``infeasible''\;
  } \Else { 
    \Return $A$\;
  }
\end{algorithm}

This leads us to considering an entire scheme of algorithms that
differ only in the lifting operation; their general scheme is shown in
Algorithm~\ref{alg:general}. All we need from this operation is that
after the lifting the matrix does not exceed the minimal solution that
is larger than the matrix before the lifting.  For instance, on the
lifting step we can add to each column the maximum of the numbers that
Akian--Gaubert--Guterman algorithm and Grigoriev's algorithm propose
to add to this column. It is easy to construct even better lifting
methods, but at present, we do not know a superpolynomial
counterexample even to this simple combination of the
Akian--Gaubert--Guterman algorithm and Grigoriev's algorithm, shown in
Algorithm~\ref{alg:combination}.  Finding such a counterexample
remains an interesting open problem that could shed light on important
properties of this class of algorithms.

\begin{algorithm}
  \label{alg:agglifting}
  \SetAlgoLined 
  \DontPrintSemicolon
  \caption{Lifting in the Akian--Gaubert--Guterman algorithm} 
  \KwData {$A$, a matrix of the tropical system} 
  \KwResult {lifted matrix} 
  \ForEach {strict minimum} { 
    find the number to add to the corresponding column such that the
    minimum ceases to be strict 
  } 
  Add to each column of $A$ the maximal of all numbers found in the
  loop.\; 
  \Return $A$\;
\end{algorithm}

\begin{algorithm}
  \SetAlgoLined 
  \DontPrintSemicolon
  \caption{Lifting in the optimized Grigoriev's
    algorithm\label{Lift2}}
  \KwData {$A$, a matrix of the tropical system}
  \KwResult {lifted matrix} 
  $J \leftarrow \emptyset$\;
  \While {there exist a row such that its minimum is achieved in a
    single column $j$ such that $j \notin J$} { 
    $J \leftarrow J \cup {j}$\; 
  } 
  $a \leftarrow \infty$\; 
  \For{$i = 0$ \KwTo $m$} { 
    $a_i \leftarrow $ maximal number that can be added to columns with
    indices from $J$ so that in the row with index $i$, minimal
    elements remain minimal\; 
    $a \leftarrow \min(a, a_i)$\; 
  }
  \ForEach {$i \in J$} { 
    add $a$ to column $i$\; 
  } 
  \Return $A$\;
\end{algorithm}

\begin{algorithm}
  \label{alg:combination}
  \SetAlgoLined 
  \DontPrintSemicolon
  \caption{Lifting for combination of Akian--Gaubert--Guterman and
    Grigoriev's algorithms}
  \KwData {$A$, a matrix of the tropical system}
  \KwResult {lifted matrix} 
  $B \leftarrow A$\;
  $C \leftarrow A$\;
  Lift $B$ with Akian--Gaubert--Guterman lifting algorithm\;
  Lift $C$ with optimized Grigoriev's lifting algorithm\;
  $A \leftarrow B \oplus C$\;
\end{algorithm}

\section{Weakly overdetermined systems}\label{sec:overdet}

In this section, we proceed to the second main result of this work,
namely an algorithm for solving overdetermined tropical systems. In
the following theorem, we show how to reduce solving an overdetermined
tropical system of width $n$ to solving $n+1$ systems corresponding to
its submatrices.  Throughout the section, we will assume that
arithmetic operations with numbers in the matrix take $O(1)$ time.

\begin{theorem}\label{thm:overdet}
Consider a matrix $A$ and $n(A) + 1$ subsets of its set of rows such
that every row of $A$ is covered at least $n$ times by these
subsets. Then, if each subset defines a feasible tropical system, $A$
is also feasible. Moreover, if solutions of each of these systems are
known, the solution of matrix $A$ can be found in polynomial time.
\end{theorem}
  \begin{proof}
We begin by constructing a matrix of solutions $S$ with rows $s_i$,
where $s_i$ is the solution for the system defined by the
$i^{\text{th}}$ subset. Then we find a solution for the tropical
system $S^\top$ ($S$ transposed), denoting it by $\alpha$; this
solution exists and can be found in polynomial time since $S^\top$ is
underdetermined: the number of equations is less than the number of
variables. Note that $s_i + \alpha_i$ is still a solution for the
system defined by the $i^{\text{th}}$ subset since multiplication by a
tropical constant preserves a solution. Consider $x =
\bigoplus_i\left(s_i+\alpha_i\right)$.  Then $x$ is the solution for
the original matrix $A$.

Indeed, since $\alpha$ is a solution, we can remove any row from this
tropical sum, and $x$ will remain unchanged (because every minimum is
achieved twice). Since each of the rows $a_i$ is covered at least $n$
times, by dropping the row corresponding to a set that does not cover
$a_i$ we get a solution for $a_i$ since the set of solutions is
linear.
  \end{proof}

\begin{example}
Let us consider how this algorithm works on the following example with an
overdetermined matrix
  \[ \left( \begin{array}{ccc}
  1 & 2 & 3 \\ 1 & 2 & 1 \\ 1 & 2 & 5 \\ 2 & 3 & 1
  \end{array} \right). \]
  We begin by choosing $4$ subsets of its rows:
  \[ \left( \begin{array}{ccc}
  1 & 2 & 3 \\ 1 & 2 & 1 \\ 1 & 2 & 5
  \end{array} \right),
  \left( \begin{array}{ccc} 1 & 2 & 3 \\ 1 & 2 & 1 \\ 2 & 3 & 1
  \end{array} \right),
  \left( \begin{array}{ccc} 1 & 2 & 3 \\ 1 & 2 & 5 \\ 2 & 3 & 1
  \end{array} \right),
  \left( \begin{array}{ccc} 1 & 2 & 1 \\ 1 & 2 & 5 \\ 2 & 3 & 1
  \end{array} \right). \]

  First we construct the matrix of solutions $S$ and its transpose:
  \[ S = \left( \begin{array}{ccc}
    1 & 0 & 0 \\ 2 & 1 & 2 \\ 3 & 2 & 3 \\ 2 & 1 & 2
  \end{array} \right), \qquad
  S^\top = \left( \begin{array}{cccc} 1 & 2 & 3 & 2 \\ 0 & 1 & 2 & 1
    \\ 0 & 2 & 3 & 2
  \end{array} \right). \]
  
  Next we solve $S^\top$, getting the solution matrix
  \[ \left( \begin{array}{cccc}
    3 & 2 & 3 & 2 \\ 2 & 1 & 2 & 1 \\ 2 & 2 & 3 & 2
  \end{array} \right) \text{ and its transpose }
  \left( \begin{array}{ccc} 3 & 2 & 2 \\ 2 & 1 & 2 \\ 3 & 2 & 3 \\ 2 &
    1 & 2
  \end{array} \right), \]
  and find the tropical sum of its rows:
  \[ x = \left( \begin{array}{ccc}
    2 & 1 & 2
  \end{array} \right). \]
  The resulting $x$ is a solution for the original matrix.
\end{example}

\begin{remark}
As it often happens in tropical mathematics, this theorem holds in the
classical case as well. Indeed, the $n + 1$ vectors that represent
solutions for subsets of equations are necessarily linearly
dependent. This means that there exists a vector $x$ that can be
expressed as a linear combination of the other vectors; by the
reasoning similar to the proof of Theorem~\ref{thm:overdet}, this
vector will be a solution for the original problem.
\end{remark}

The most straightforward way to turn Theorem~\ref{thm:overdet} into an
algorithm is to choose on each step $n+1$ subsets with $n + k - 1$
rows each, making sure that for each of the sets the absent rows are
different. Thus, the feasibility problem for the original matrix can
be reduced to $n + 1$ problems of smaller size. To estimate the
complexity of the resulting algorithm in terms of $n$ and $k$, we
denote this complexity by $T(n,k)$. Recall that $T(n, 0) =
poly(n)$. We get the following recurrent relation for $T(n, k)$: $T(n,
k) = (n + 1)T(n, k - 1) + poly(n)$. This means that $T(n, k) = (n +
1)^k poly(n)$, which is a polynomial for $k$ bounded by a constant, so
we have arrived at the following theorem.

\begin{theorem}
The problem of solving weakly overdetermined tropical linear systems
(systems for which $k$ is bounded by a constant) can be solved in
polynomial time.
\end{theorem}

One can consider other ways of choosing the subsets. One of the most
efficient methods is the following: consider a matrix $A$ for which we
need to find a solution. First, we introduce an ordering on the rows
of a matrix corresponding to the order of rows in the matrix $A$; we
will further assume that rows in all subsets are ordered in this
way. On each step, we choose $n+1$ subsets as follows: the first $n$
rows, all rows except the first, all rows except the second, and so
on, ending with all rows except the $n^{\text{th}}$.  To further
improve the algorithm's running time, we use dynamical programming,
storing the matrices that appear over the course of the algorithm's
operation and their corresponding solutions in order to reuse them if
the same matrix appears for a second time.

After this optimization, all we need to estimate the running time is
to estimate the number of submatrices appearing in the algorithm.
Note that all submatrices look as follows: the first $n$ rows are an
arbitrary ordered subset of rows, and the rest are always several
consecutive last rows. Thus, the number of matrices can be bounded
from above by $m{m \choose n}$ (there are ${m \choose n}$ ways to
choose the first $n$ rows, and the other rows are uniquely defined by
their number). As a result, we get that the algorithm for solving
weakly overdetermined systems has complexity ${m \choose n} poly(m,
n)$, which up to a polynomial coincides with the upper bound for the
best known algorithm for feasibility testing in tropical linear
systems. This bound is better than the upper bound on Grigoriev's
algorithm but worse than the best known lower bound. With this
algorithm, outlined in Algorithm~\ref{alg:over}, we have finally
proven our main result in this section.

\begin{theorem}
Any tropical linear system can be solved in time $$\Theta\left({m
  \choose n} \mathrm{poly}\left(m, n, \log M\right)\right).$$
\end{theorem}

\begin{algorithm}\label{alg:over}
  \SetAlgoLined 
  \DontPrintSemicolon
  \caption{Algorithm for solving weakly overdetermined systems}
  \KwData {$A$, a matrix of the tropical system} 
  \KwResult {if the tropical system defined by $A$ was feasible then
    its solution matrix, else ``infeasible''} 
  \tcp{solutions for submatrices are obtained with this algorithm}

  \If {$n(A) = m(A)$} { 
    \Return{a solution for the matrix $A$ obtained by Grigoriev's
      algorithm for solving square tropical matrices}\; 
  } 
  \If {$n(A) > m(A)$} { 
    \Return{a solution for the matrix $A$ obtained with tropical
      Cramer's rule}\; 
  } 
  $S \leftarrow \emptyset$\; 
  \uIf { $\{a_1, a_2, ..., a_n\}$ is feasible} { 
    $S \leftarrow S \cup $ \{solution of $\{a_1, a_2, ..., a_n\}$\}\;
  } \Else { 
    \Return ``infeasible''\; 
  } \For {$i = 0$ \KwTo $n$} { 
    \uIf { $\{a_1, a_2, ..., a_{i - 1}, a_{i + 1}, ..., a_n\}$ is
      feasible} { 
      $S \leftarrow S \cup $ \{solution $\{a_1, a_2, ..., a_{i - 1},
      a_{i + 1}, ..., a_n\}$\}\; 
    }
    \Else { 
      \Return ``infeasible''\; 
    } 
  } 
  $B \leftarrow$ matrix of the rows contained in the set $S$\; 
  Transpose $B$\; 
  $C \leftarrow$ solution matrix of system $B$ obtained with tropical
  Cramer's rule\; 
  Transpose $C$\; 
  \Return{tropical sum of the rows of $C$}\;
\end{algorithm}

\section{Conclusion}\label{sec:conclusion}

In this work, we have presented new algorithms for solving tropical
linear systems: a modification of Grigoriev's algorithm that leads to
a new family of algorithms with different lifting operations and a
novel algorithm for solving overdetermined tropical systems that has
the same time complexity as the previously known feasibility testing
algorithm.

Tropical linear systems turn out to have very interesting
computational properties: the problem lies in $NP\cap coNP$ but no
polynomial algorithm is known, and the best algorithm known so
far was polynomial in any two of its characteristics out of three
(width, height, and maximal element in the system's matrix). Further
work in this direction may include further improvements of the
algorithms proposed in this paper and finding counterexamples for the
new algorithms proposed here: while we do claim that we have made
Grigoriev's algorithm significantly faster in practice, we doubt that
a simple combination of Grigoriev's algorithm and the
Akian--Gaubert--Guterman is indeed polynomial, so we expect the search
for hard counterexamples to succeed.

Another interesting direction for further study comes from the notion
that for every linear set of points there is a minimal tropical
prevariety (with respect to inclusion) that contains this set. In low
dimensions such a prevariety can be constructed as a closure of the
original set under several simple operations, and one can test
infeasibility by testing if a prevariety built on a set of vectors
contains the entire space.  We believe that this approach may lead to
new algorithms for solving tropical linear systems.

\bibliography{tropical}

\begin{thebibliography}{10}

\bibitem{Akian10}
M.~{Akian}, S.~{Gaubert}, and A.~{Guterman}.
\newblock The correspondence between tropical convexity and mean payoff games.
\newblock {\em Proc. 19 Intern. Symp. Math. Theory of Networks and Systems},
  pages 1295--1302, 2010.

\bibitem{Akian12}
M.~{Akian}, S.~{Gaubert}, and A.~{Guterman}.
\newblock Tropical polyhedra are equivalent to mean payoff games.
\newblock {\em International Journal of Algebra and Computation}, 22(1), 2012.

\bibitem{Bezem10}
M.~{Bezem}, R.~{Nieuwenhuis}, and {Rodrígez-Carbonell} E.
\newblock Hard problems in max-algebra, control theory, hypergraphs and other
  areas.
\newblock {\em Inf. Process. Lett.}, 110(4):133--138, 2010.

\bibitem{Butcovic94}
P.~{Butcovič}.
\newblock Strong regularity of matrices — a survey of results.
\newblock {\em Discrete Applied Mathematics}, 48:45--68, 1994.

\bibitem{Cuninghame-Green79}
R.~A. {Cuninghame-Green}.
\newblock Minimax algebra and applications.
\newblock {\em Springer Lect. Notes in Economics and Mathematical Systems},
  166, 1979.

\bibitem{Davydow13eng}
A.~P. Davydow.
\newblock Upper and lower bounds for grigoriev’s algorithm for solving
  integral tropical linear systems.
\newblock {\em Journal of Mathematical Sciences}, 192(3):295--302, 2013.

\bibitem{Develin05}
M.~{Develin}, F.~{Santos}, and {Sturmfels} B.
\newblock On the rank of a tropical matrix.
\newblock {\em Math. Sci. Res. Inst. Publ., Combinatorial and computational
  geometry}, 52:213--242, 2005.

\bibitem{Erenfeucht79}
A.~{Ehrenfeucht} and J.~{Mycielski}.
\newblock Positional strategies for mean payoff games.
\newblock {\em International Journal of Game Theory}, 8:109--113, 1979.

\bibitem{Floyd62}
Robert~W. Floyd.
\newblock Algorithm 97: Shortest path.
\newblock {\em Commun. ACM}, 5(6):345--, June 1962.

\bibitem{Grigoriev-personal}
D.~{Grigoriev}.
\newblock Personal communication.

\bibitem{Grigoriev13}
D.~{Grigoriev}.
\newblock Complexity of solving tropical linear systems.
\newblock {\em Comput. Complexity}, 22:71--78, 2013.

\bibitem{Grigoriev12}
D.~{Grigoriev} and V.~V. {Podoltkii}.
\newblock Complexity of tropical and min-plus linear prevarieties.
\newblock {\em CoRR, abs/1204.4578}, 2012.

\bibitem{Izhakian08}
Z.~{Izhakian}.
\newblock The tropical rank of a tropical matrix.
\newblock {\em Eprint arXiv:math.AC/0604208v2}, 2008.

\bibitem{Izhakian09}
Z.~{Izhakian} and L.~{Rowen}.
\newblock The tropical rank of a tropical matrix.
\newblock {\em Communications in Algebra}, 37(11):3912--3927, 2009.

\bibitem{Karp78}
R.~M. {Karp}.
\newblock A characterization of the minimum cycle mean in a digraph.
\newblock {\em Discrete Mathematics}, 23:309--311, 1978.

\bibitem{Karzanov88}
A.~V. {Karzanov}, V.~A. {Gurvich}, and L.~G. {Khaciyan}.
\newblock Cyclic games and an algorithm to finnd minimax cycle means in
  directed graphs.
\newblock {\em USSR Computational Mathematics and Mathematical Physics},
  28:85--91, 1988.

\bibitem{Karzanov93}
A.~V. {Karzanov} and V.~N. {Lebedev}.
\newblock Cyclical games with prohibitions.
\newblock {\em Mathematical Programming}, 60:277--293, 1993.

\bibitem{Kleene56}
S.~C. {Kleene}.
\newblock Representation of events in nerve nets and finite automata.
\newblock {\em Automata Studies, ed. by C.E. Shannon and J. McCarthy. Annals of
  Mathematics Studies}, 34, 1956.

\bibitem{Krivulin09}
N.~K. {Krivulin}.
\newblock Methods of idempotent algebra in problems of complex systems modeling
  and analysis.
\newblock {\em St. Petersburg Univ. Press}, 2009.

\bibitem{Kuhn55}
Harold~W. {Kuhn}.
\newblock The hungarian method for the assignment problem.
\newblock {\em Naval Research Logistics Quarterly}, 2:83--97, 1955.

\bibitem{Litvinov07}
G.~L. {Litvinov}.
\newblock {The Maslov dequantization, idempotent and tropical mathematics: A
  brief introduction}.
\newblock {\em Journal of Mathematical Science}, 140(3), July 2007.

\bibitem{Litvinov11}
G.~L. {Litvinov}, V.~P. {Maslov}, {Rodionov}~A. Ya., and A.~N. {Sobolevski}.
\newblock Universal algorithms, mathematics of semirings and parallel
  computations.
\newblock {\em Lecture Notes in Computational Science and Engineering}, 75,
  2011.

\bibitem{Lozovanu91}
D.~D. {Lozovanu}.
\newblock Algorithms to solve some classes of network minimax problems and
  their applications.
\newblock {\em Cybernetics}, 29:93--100, 1991.

\bibitem{Lozovanu93}
D.~D. {Lozovanu}.
\newblock Strongly polynomial algorithms for nding minimax paths in networks
  and solution of cyclic games.
\newblock {\em Cybernetics and Systems Analysis}, 29:754--759, 1993.

\bibitem{Mikhalkin03}
G.~{Mikhalkin}.
\newblock {Enumerative tropical algebraic geometry in R2}.
\newblock {\em ArXiv Mathematics e-prints}, December 2003.

\bibitem{NGVR12b}
V.~Noel, D.~Grigoriev, S.~Vakulenko, and O.~Radulescu.
\newblock Hybrid models of the cell cycle molecular machinery.
\newblock {\em Electronic Proceedings in Theoretical Computer Science},
  92:88--105, 2012.

\bibitem{NGVR12a}
V.~Noel, D.~Grigoriev, S.~Vakulenko, and O.~Radulescu.
\newblock Tropical geometries and dynamics of biochemical networks. application
  to hybrid cell cycle models.
\newblock {\em SASB 2011, Electronic Notes in Theoretical Computer Science},
  284:75--91, 2012.

\bibitem{Pin98}
Jean-Eric Pin.
\newblock Positive varieties and infinite words.
\newblock In {\em LATIN}, pages 76--87, 1998.

\bibitem{Gebert03}
J.~{Richter-Gebert}, B.~{Sturmfels}, and T.~{Theobald}.
\newblock {First steps in tropical geometry}.
\newblock {\em ArXiv Mathematics e-prints}, June 2003.

\bibitem{Imre88}
Imre Simon.
\newblock Recognizable sets with multiplicities in the tropical semiring.
\newblock In MichalP. Chytil, Václav Koubek, and Ladislav Janiga, editors,
  {\em Mathematical Foundations of Computer Science 1988}, volume 324 of {\em
  Lecture Notes in Computer Science}, pages 107--120. Springer Berlin
  Heidelberg, 1988.

\bibitem{Sturmfels04}
D.~{Speyer} and B.~{Sturmfels}.
\newblock {Tropical Mathematics}.
\newblock {\em ArXiv Mathematics e-prints}, August 2004.

\bibitem{Zwick95}
U.~{Zwick} and M.~{Paterson}.
\newblock The complexity of mean payoff games on graphs.
\newblock {\em Theoret. Comput. Sci.}, 158(1-2):342--359, 1995.

\end{thebibliography}
\end{document}